\documentclass[letterpaper, 10 pt, conference]{IEEEtran}

\usepackage{cite}

\usepackage[cmex10]{amsmath}
\usepackage{amssymb,amsfonts}
\usepackage[hyphens]{url}
\usepackage{algorithmic,algorithm}
\usepackage{amsthm}

\usepackage{graphicx}
\usepackage[caption=false,font=footnotesize]{subfig}
\usepackage{multirow}
\usepackage{etoolbox}

\newtoggle{isreport}
\toggletrue{isreport}

\newtoggle{isarxiv}
\toggletrue{isarxiv}

\newtheorem{lemma}{Lemma}

\newtheorem{prop}[lemma]{Proposition}

\newtheorem{thm}[lemma]{Theorem}
\theoremstyle{theorem}\newtheorem{eg}{Example}
\theoremstyle{theorem}\newtheorem{defn}[lemma]{Definition}

\newcommand{\qedd}{\hfill{$\square$}}

\newcommand{\paren}[1]{\left(#1\right)}

\newcommand{\set}[1]{\left\{#1\right\}}
\newcommand{\abs}[1]{\left|#1\right|}

\newcommand{\normm}[1]{{\left\vert\kern-0.25ex\left\vert\kern-0.25ex\left\vert #1
    \right\vert\kern-0.25ex\right\vert\kern-0.25ex\right\vert}}

\DeclareMathOperator{\diag}{diag}

\newcommand{\R}{\mathbb R}

\newcommand{\ol}[1]{\overline{#1}}

\newcommand{\bs}{\backslash}

\newcommand{\calB}{\mathcal{B}}
\newcommand{\calC}{\mathcal{C}}

\newcommand{\calE}{\mathcal{E}}

\newcommand{\calG}{\mathcal{G}}

\newcommand{\calI}{\mathcal{I}}

\newcommand{\calL}{\mathcal{L}}

\newcommand{\calN}{\mathcal{N}}

\newcommand{\calP}{\mathcal{P}}

\newcommand{\calT}{\mathcal{T}}

\usepackage{color}
\usepackage{enumerate}

\begin{document}
\title{Failure Localization in Power Systems via Tree Partitions}

\IEEEoverridecommandlockouts

\author{Linqi Guo, Chen~Liang, Alessandro~Zocca, Steven H.~Low,~and~Adam~Wierman
\thanks{This work has been supported by Resnick Fellowship, Linde Institute Research Award, NWO Rubicon grant 680.50.1529., NSF grants through PFI:AIR-TT award 1602119, EPCN 1619352, CNS 1545096, CCF 1637598, ECCS 1619352, CNS 1518941, CPS 154471, AitF 1637598, ARPA-E grant through award DE-AR0000699 (NODES) and GRID DATA, DTRA through grant HDTRA 1-15-1-0003 and Skoltech through collaboration agreement 1075-MRA.}
\thanks{The authors are with the Department of Computing and Mathematical Sciences, California Institute of Technology, Pasadena,
CA, 91125, USA. Email: \texttt{\{lguo, cliang2, azocca, slow, adamw\}@caltech.edu}.}}

\maketitle

\begin{abstract}
Cascading failures in power systems propagate non-locally, making the control and mitigation of outages extremely hard. In this work, we use the emerging concept of the \emph{tree partition} of transmission networks to provide an analytical characterization of line failure localizability in transmission systems. Our results rigorously establish the well perceived intuition in power community that failures cannot cross bridges, and reveal a finer-grained concept that encodes more precise information on failure propagations within tree-partition regions. Specifically, when a non-bridge line is tripped, the impact of this failure only propagates within well-defined components, which we refer to as \emph{cells}, of the tree partition defined by the bridges. In contrast, when a bridge line is tripped, the impact of this failure propagates globally across the network, affecting the power flow on all remaining transmission lines. This characterization suggests that it is possible to improve the system robustness by \emph{temporarily} switching off certain transmission lines, so as to create more, smaller components in the tree partition; thus spatially localizing line failures and making the grid less vulnerable to large-scale outages. We illustrate this approach using the IEEE 118-bus test system and demonstrate that switching off a negligible portion of transmission lines allows the impact of line failures to be significantly more localized without substantial changes in line congestion.  
\end{abstract}

\section{Introduction}\label{section:intro}
Power system reliability is a crucial component in the development of sustainable modern power infrastructure. Recent blackouts, especially the 2003 and 2012 blackouts in Northwestern U.S.~\cite{report2003blackout} and India \cite{report2013blackout}, demonstrated the devastating economic impact a grid failure can cause. In even worse cases, where facilities like hospitals are involved, blackouts pose direct threat to people's health and lives.

Because of the intricate interactions among power system components, outages may cascade and propagate in a very complicated, non-local manner~\cite{hines2017cascading, dobson2016statistics,bernstein2014power}, exhibiting very different patterns for different networks~\cite{soltan2015analysis}. Such complexity originates from the interplay between network topology and power flow physics, and is aggravated by possible hidden failures \cite{chen2005cascading} and human errors~\cite{carreras2002critical}.  This complexity is the key challenge for research into the modeling, control, and mitigation of cascading failures in power systems.  

There are three traditional approaches for  characterizing the behavior of cascades in the literature: (i) using simulation models  \cite{dobson2007complex} that rely on Monte-Carlo approaches to account for the steady state power flow redistribution on DC \cite{carreras2002critical,anghel2007stochastic,yan2015cascading,bernstein2014power} or AC \cite{nedic2006criticality,rios2002value,song2016dynamic} models; (b) studying purely topological models that impose certain assumptions on the cascading dynamics (e.g., failures propagate to adjacent lines with high probability) and infer component failure propagation patterns from graph-theoretic properties \cite{brummitt2003cascade,kong2010failure,crucitti2004topological}; (c) investigating simplified or statistical cascading failure dynamics \cite{dobson2005probilistic,wang2012markov,rahnamay2014stochastic,hines2017cascading}. In each of these approaches, it is typically challenging to make general inferences across different scenarios due to the lack of structural understanding of power redistribution after line failures.

A new approach has emerged in recent years, which seeks to use spectral properties of the network graph in order to derive precise structural properties of the power system dynamics, e.g., \cite{guo2017spectral,guo2017monotonicity,guo2018graph,guo2018cyber}.  The spectral view is powerful as it often reveals surprisingly simple characterizations of the complicated system behaviors. In the cascading failure context, a key result from this approach is about the \emph{line outage distribution factor} \cite{wood1996generation,soltan2015analysis}. Specifically, it is shown in \cite{guo2017monotonicity} that the line outage distribution factor is closely related to transmission graph spanning forests. 

While this literature has yet to yield a precise characterization of cascades, it has suggested a new structural representation of the transmission graph called the \emph{tree partition}, which is particularly promising.  For example, \cite{guo2017monotonicity} shows that line failures in a transmission system cannot propagate across different regions of the tree partition (for more background on the tree partition, see Section \ref{section:basic}).

\textbf{Contributions of this paper:} \emph{We prove that the tree partition proposed in \cite{guo2017monotonicity} can be used to provide an analytical characterization of line failure localizability, under a DC power flow model, and we show how to use this characterization to mitigate failure cascades by temporarily switching off a small number of transmission lines.} Our results rigorously establish the well perceived intuition in power community that failures cannot cross bridges, and reveal a finer-grained concept that encodes more precise information on failure propagations within tree-partition regions. This work builds on the recent work focused on the line outage distribution factor, e.g., \cite{lai2013allerton,soltan2015analysis,guo2017monotonicity}, and shows that the tree partition is a particularly useful representation of this factor, one that captures many aspects of how line failures can cascade.

The formal characterization of localizability is given in Theorem \ref{thm:summary}, which summarizes the technical results in Sections \ref{section:non-bridge} and \ref{section:bridge}.  In particular, in Section \ref{section:non-bridge}, we characterize the power redistribution after the tripping of a non-bridge line and show that the impact of such failures only propagates within well-defined components, which we refer to as cells, inside the tree partition regions. In Section \ref{section:bridge}, we consider the failure of bridge lines and prove that, in normal operating conditions, such failures propagate globally across the network and impact the power flow on all transmission lines. 
In order to prove these results, we depend on properties of the tree partition proved in \cite{guo2017monotonicity} as well as some novel properties derived in Section \ref{section:basic}. Further, we make use of the block decomposition of tree partition regions to completely eliminate the graph spanning forests among distinct cells, which in the spectral view means failure localization~\cite{guo2017monotonicity}. Lastly, we apply classical techniques from algebraic geometry to address potential pathological system specifications and establish our results.

The characterization we provide in Theorem \ref{thm:summary} yields many interesting insights for the planning and management of power systems and, further, suggests a new approach for mitigating the impact of cascading failures. Specifically, our characterization highlights that switching off certain transmission lines \emph{temporarily} in responds to the real-time injection profile can lead to more, smaller regions/cells, which localize line failures, thus making the grid less vulnerable against line outages. In Section \ref{section:case_study}, we illustrate this approach using the IEEE 118-bus test system. We demonstrate that switching off only a negligible portion of transmission lines can lead to significantly better control of cascading failures.  Further, we highlight that this happens without significant increases in line congestion across the network. 

\section{Preliminaries}\label{section:model}
We use the graph $\calG=(\calN,\calE)$ to describe a power transmission network, where $\calN=\set{1,\ldots, n}$ is the set of buses and $\calE\subset\calN \times \calN$ is the set of transmission lines. The terms bus/vertex and line/edge are used interchangeably. An edge in $\calE$  between vertices $i$ and $j$ is denoted either as $e$ or $(i,j)$. We assume $\calG$ is connected and simple, and assign an arbitrary orientation over $\calE$ so that if $(i,j)\in\calE$ then $(j,i)\notin\calE$. The line susceptance of $e$ is denoted as $B_e$ and the branch flow on $e$ is denoted as $P_e$. The susceptance matrix is defined to be the diagonal matrix $B=\diag(B_e: e\in\calE)$. 

We denote the power injection and phase angle at bus $i$ as $p_i$ and $\theta_i$, and use $n$ and $m$ to denote the number of buses and transmission lines in $\calG$. The vertex-edge incidence matrix of $\calG$ is the $n\times m$ matrix $C$ defined as
$$
C_{ie}=\begin{cases}
  1 & \text{if vertex }i\text{ is the source of }e\\
  -1 & \text{if vertex }i\text{ is the target of }e\\
  0 &\text{otherwise.}
\end{cases}
$$
With the above notation, the DC power flow model can be written as
\begin{subequations}\label{eqn:dc_model}
\begin{IEEEeqnarray}{rCl}
	p&=&CP \label{eqn:flow_conservation}\\
	P&=&BC^T\theta, \label{eqn:kirchhoff}
\end{IEEEeqnarray}
\end{subequations}
where \eqref{eqn:flow_conservation} is the flow conservation constraint and \eqref{eqn:kirchhoff} is Kirchhoff's and Ohm's Laws. The slack bus phase angle in $\theta$ is typically set to $0$ as a reference to other buses. With this convention, the DC model \eqref{eqn:dc_model} has a unique solution $\theta$ and $P$ for each injection vector $p$ such that $\sum_{j\in\calN}p_j=0$.

\begin{figure}
\centering
\iftoggle{isarxiv}{
\includegraphics[width=.275\textwidth]{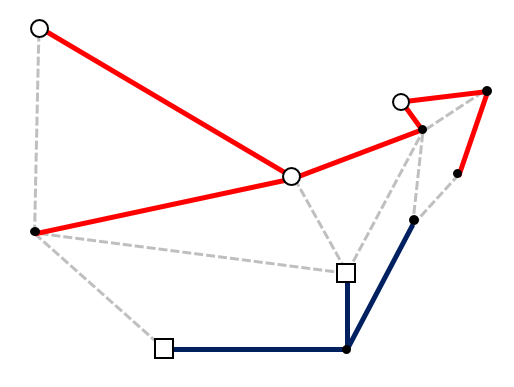}
}{
\includegraphics[width=.275\textwidth]{figs/spanning_forest.png}
}
\caption{An example element in $\calT(\calN_1,\calN_2)$, where circles correspond to elements in $\calN_1$ and squares correspond to elements in $\calN_2$. The two trees containing $\calN_1$ and $\calN_2$ are highlighted as solid lines.}\label{fig:spanning_forest}
\end{figure}

When a line $e$ is tripped, the power flow redistributes according to the DC model \eqref{eqn:dc_model} on the newly formed graph $\calG'=(\calN,\calE\bs\set{e})$. If $\calG'$ is still connected, then the branch flow change on a line $\hat{e}$ is given as
$$
	\Delta P_{\hat{e}} = P_e\times K_{e\hat{e}},
$$
where $K_{e\hat{e}}$ is the \emph{line outage distribution factor} \cite{wood1996generation} from $e$ to $\hat{e}$. It is known that this distribution factor is independent of the original power injection $p$ and can be computed from the matrices $B$ and $C$ \cite{wood1996generation}.

If the new graph $\calG'$ is disconnected, then it is possible that the original injection $p$ is no longer balanced in the connected components of $\calG'$. Thus, to compute the new power flow, a certain power balance rule $\calB$ needs to be applied. Several such rules have been proposed and evaluated in literature based on load shedding or generator response \cite{soltan2015analysis,bernstein2014power,bienstock2010n-k,bienstock2011control,bernstein2014vulnerability}. In this work, we do not specialize to any such rule and instead opt to identify the key properties of these rules that allow our results to hold. With this more abstract approach, we can characterize the power flow redistribution under a class of power balance rules. 

The line outage distribution factor has been extensively studied in previous work, and we make use of some important properties from this literature. In \cite{guo2017monotonicity}, by studying the spectrum of the graph Laplacian, a formula for $K_{e\hat{e}}$ is given in terms of the graph structure of $\calG$. This formula plays a central role in proving almost all of the results in this paper, and we thus restate it here. To do so, we need some more notation. For two subsets of vertices $\calN_1, \calN_2\subset\calN$, let $\calT(\calN_1, \calN_2)$ be the set of spanning forests of $\calG$ consisting of \emph{exactly} two trees that contain $\calN_1$ and $\calN_2$ respectively\footnote{This definition does not require $\calN_1$ and $\calN_2$ to be disjoint, and $\calT(\calN_1, \calN_2)$ is understood to be empty when $\calN_1\cap\calN_2\neq \emptyset$.  See \cite{guo2017monotonicity} for more discussions.}. See Fig.~\ref{fig:spanning_forest} for an illustration. Given a set $E\subset\calE$ of lines, we assign a weight to $E$ by multiplying the susceptances over all lines in $E$
$$
\chi(E):=\prod_{e\in E}B_e
$$
and denote the set of all spanning trees of $\calG$ with edges only from $E$ by $\calT_E$ (which can be empty if $E$ is too small). The following result from \cite{guo2017monotonicity} is used throughout this paper.
\begin{thm}\label{thm:redist_factor}
Let $e=(i,j), \hat{e}=(w,z)$ be edges such that the $\calG'=(\calN, \calE\bs\set{e})$ is connected. Then $K_{e\hat{e}}$ is given by
\begin{equation*}
B_{\hat{e}}\times\frac{\sum_{E\in \calT(\set{i,w},\set{j,z})}\chi(E)-\sum_{E\in \calT(\set{i,z},\set{j,w})}\chi(E)}{\sum_{E\in\calT_{\calE\bs \set{(i,j)}}}\chi(E)}
\end{equation*}
\end{thm}
This result reflects Kirchhoff's Law in a precise way and shows that the impact of a line failure propagages through spanning forests inside $\calG$. Interested readers are directed to \cite{guo2017monotonicity} for a more detailed discussion.

\section{Tree Partition: Definition and Properties}\label{section:basic}
\begin{figure}
\centering
\iftoggle{isarxiv}{
\includegraphics[width=0.325\textwidth]{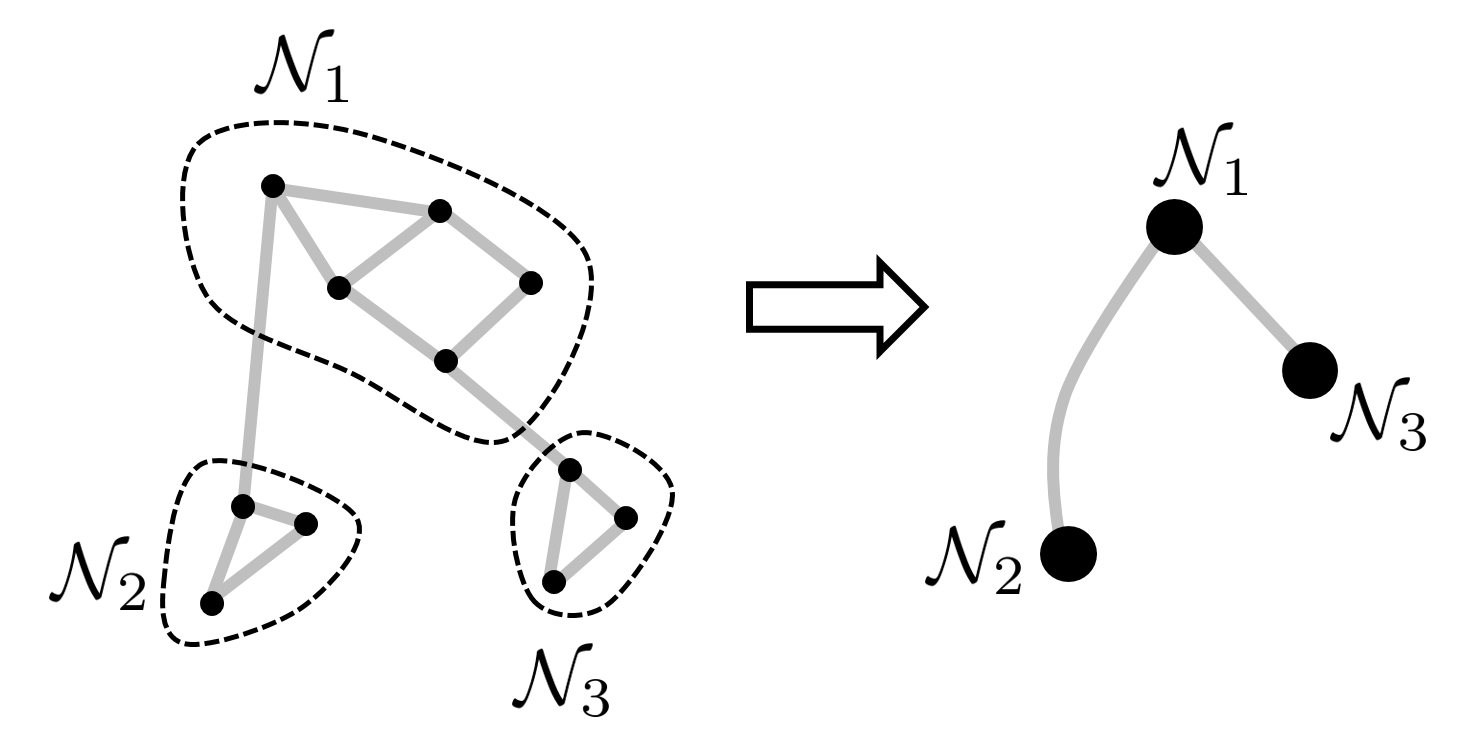}
}{
\includegraphics[width=0.325\textwidth]{figs/tree_partition.png}
}
\caption{The construction of $\calG_\calP$ from $\calP$.}\label{fig:tree_partition}
\end{figure}

The idea of tree partition for analyzing cascading failure was first introduced in \cite{guo2017monotonicity}. Here, we define the tree partition, discuss its uniqueness and show that the ``finest'' tree partition of a general graph can be computed in linear time.

For a power network $\calG=(\calN,\calE)$, a collection $\calP=\set{\calN_1,\calN_2,\cdots, \calN_k}$ of subsets of $\calN$ is said to form a \emph{partition} of $\calG$ if $\calN_i\cap\calN_j=\emptyset$ for $i\neq j$ and $\cup_{i=1}^k \calN_i=\calN$. For any partition, we can define a reduced multi-graph $\calG_\calP$ from $\calG$ as follows. First, we reduce each subset $\calN_i$ to a super node (see Fig.~\ref{fig:tree_partition}). The collection of all super nodes forms the node set for $\calG_\calP$. Second, we add an undirected edge connecting the super nodes $\calN_i$ and $\calN_j$ for each pair of $n_i, n_j\in\calN$ with the property that $n_i\in\calN_i$, $n_j\in\calN_j$ and $n_i$ and $n_j$ are connected in $\calG$. Note that multiple ledges are added when multiple pairs of such $n_i,n_j$ exist. Unlike the graph $\calG$ to which we assign an arbitrary orientation (and thus is a directed graph), the reduced multi-graph $\calG_\calP$ is undirected.

\begin{defn}
A partition $\calP=\set{\calN_1,\calN_2,\cdots,\calN_k}$ of $\calG$ is said to be a \textbf{tree partition} if the reduced graph $\calG_\calP$ forms a tree.
\end{defn}

\begin{defn}
Given a tree partition $\calP=\set{\calN_1,\calN_2,\cdots,\calN_k}$, the sets $\calN_i$ are called the \textbf{regions} of $\calP$. An edge $e=(w,z)$ with both endpoints inside $\calN_i$ is said to be \textbf{within} $\calN_i$. If $e$ is not within $\calN_i$ for any $i$, then we say $e$ forms a \textbf{bridge}\footnote{We remark that our definition of bridges agrees with the classical definition of bridges in graph theory (i.e., the removal of any such edge disconnects the original graph) in the sense that if the tree partition $\calP$ is irreducible (see Definition \ref{defn:irreducible} later) any bridge defined in our sense is a bridge in the classical sense, and vice versa.}.
\end{defn}

Tree partitions of a power network $\calG$ are generally not unique. For instance, one can always collapse $\calG$ into a single region with the partition $\calP_0=\set{\calN}$, which is a trivial tree partition of $\calG$. This yields a different tree partition for the graph shown in Fig.~\ref{fig:tree_partition}. Nevertheless, if we require the tree partition to be as ``fine'' as possible, such a partition is unique.

More concretely, given a graph $\calG$, we define a partial order $\succeq$ over the set of all tree partitions of $\calG$ (which is nonempty as it always contains the trivial partition $\calP_0$) as follows: For two tree partitions $\calP^1=\set{\calN_1^1,\calN_2^1,\cdots,\calN_{k_1}^1}$ and $\calP^2=\set{\calN_1^2,\calN_2^2,\cdots,\calN_{k_2}^2}$, we say $\calP^1$ is \emph{finer} than $\calP^2$, denoted as $\calP^1\succeq \calP^2$, if for any $i=1,2,\ldots, k_1$, there exists some $j(i)\in \set{1,2,\ldots, k_2}$ such that $\calN_i^1\subset \calN_{j(i)}^2$. That is, $\calP^1$ is finer than $\calP^2$ if each region in $\calP^1$ is contained in some region in $\calP^2$ (see Fig.~\ref{fig:partition_order}). It is routine to check that $\succeq$ defines a partial order over all possible tree partitions of $\calG$.

\begin{defn}\label{defn:irreducible}
A tree partition $\calP$ of $\calG$ is said to be \textbf{irreducible} if $\calP$ is maximal with respect to the partial order $\succeq$.
\end{defn}

In other words, an irreducible tree partition $\calP$ of $\calG$ is a partition that cannot be reduced to a finer tree partition.

\begin{prop}\label{prop:unique_irreducible}
For any graph $\calG$, there exists a unique irreducible tree partition.
\end{prop}
\iftoggle{isreport}{
See Appendix \ref{proof:uniqueness} for a proof.
}{
See our online report for a proof \cite{report}.
}

\begin{figure}[t]
\centering
\iftoggle{isarxiv}{
\includegraphics[width=.275\textwidth]{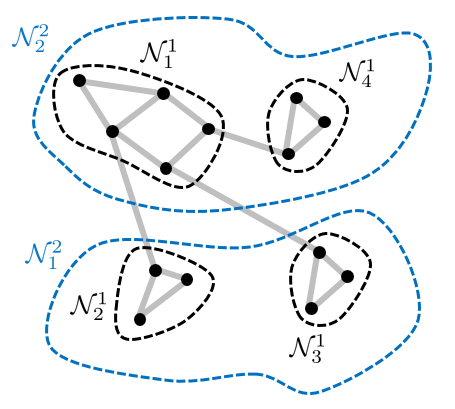}
}{
\includegraphics[width=.275\textwidth]{figs/partition_order.png}
}
\caption{An illustration of the partial order $\succeq$ over tree partitions. The partition $\calP^1=\set{\calN_1^1, \calN_2^2, \calN_3^1, \calN_4^1}$ is finer than $\calP^2=\set{\calN_1^2, \calN_2^2}$.}\label{fig:partition_order}
\end{figure}
We remark that our proof of Proposition \ref{prop:unique_irreducible} not only shows that the irreducible tree partition of $\calG$ is unique, but also implies that the problem of computing this unique irreducible tree partition reduces to finding all bridges of $\calG$. As a result, we can adapt Tarjan's bridge-finding algorithm \cite{tarjan1974note} to devise an algorithm that computes the irreducible tree partition of $\calG$ in $O(n+m)$ time. This is summarized in Algorithm \ref{alg:tree_partition}. Interested readers are referred to the proof of Proposition \ref{prop:unique_irreducible} in\iftoggle{isreport}{ Appendix \ref{proof:uniqueness}}{~\cite{report}} for more details on the algorithm.

\begin{algorithm}
\caption{Irreducible Tree Partition Finding Algorithm}\label{alg:tree_partition}
\begin{algorithmic}[1]
    \STATE Execute Tarjan's bridge-finding algorithm \cite{tarjan1974note} on $\calG=(\calN,\calE)$ to compute the set of bridges $\calE_b$.
    \STATE Remove edges in $\calE_b$ from $\calE$ to form the partitioned graph $(\calN, \calE\bs\calE_b)$.
    \STATE Breadth-first search on the partitioned graph $(\calN,\calE\bs\calE_b)$ to compute its set of connected components $\calP:=\set{C_1, C_2, \ldots, C_k}$. Return $\calP$.
     \end{algorithmic}
\end{algorithm}

To summarize, we have shown that each graph $\calG$ has a unique irreducible tree partition, which can be computed in linear time. Thus, to simplify the terminology, whenever we say the tree partition of $\calG$ in the sequel, we always refer to its irreducible partition. 

\section{Summary of Results}\label{section:summary}
In this section we state our main result, which analytically characterizes line failure localization.  It summarizes the technical results in the two sections that follow.

Our main result applies in contexts where the system is operating under \emph{normal} conditions, i.e., when the following two assumptions are satisfied: (a) the injection is \textit{island-free} (see Definition \ref{defn:island_free} for a formal definition); and
(b) the grid is \textit{participating} with respect to its power balance rule (see Definition \ref{defn:participating} for a formal definition). Moreover, to address certain  pathological cases, we add a perturbation drawn from certain probability measure $\mu$ to the line susceptances and assume $\mu$ is absolutely continuous with respect to the Lebesgue measure $\calL_m$ on $\R^m$ (see Section \ref{section:bridge}). 

\begin{thm}\label{thm:summary}
For a power network operating under normal conditions, $K_{e\hat{e}}\neq 0$ almost surely in $\mu$ if and only if:
\begin{enumerate}
\item $e,\hat{e}$ are within a common tree partition region and $e,\hat{e}$ belong to the same cell; or
\item $e$ is a bridge.
\end{enumerate}
\end{thm}

This result highlights that, for a practical system, the tree partition encodes rich information on how the failure of a line propagates through the network. We emphasize that: (i) the condition that $\mu$ is absolutely continuous with respect to $\calL_m$ is satisfied by almost all practical probability models for such perturbations (see Section \ref{section:non-bridge}); and (ii) the conditions that the injection is island-free and the grid is participating are satisfied in typical operating scenarios (see Section \ref{section:bridge}). Therefore, the conditions posed in Theorem \ref{thm:summary} are satisfied in practical settings.


Fig.~\ref{fig:matrix_block} shows how the tree partition is linked to the sparsity of the $K_{e\hat{e}}$ matrix through Theorem \ref{thm:summary}. It suggests that, compared to a full mesh transmission network consisting of single region/cell, it can be beneficial to \emph{temporarily} switch off certain lines so that more regions/cells are created and the impact of a line failure is localized within the cell in which the failure occurs. We study this network planning and design opportunity in Section~\ref{section:case_study}.

In the next two sections we prove Theorem \ref{thm:summary}.  We first characterize the power redistribution after the tripping of a non-bridge line in Section \ref{section:non-bridge}, and then consider the failure of bridge lines in Section \ref{section:bridge}.

\begin{figure}[t]
\centering
\iftoggle{isarxiv}{
\includegraphics[width=.25\textwidth]{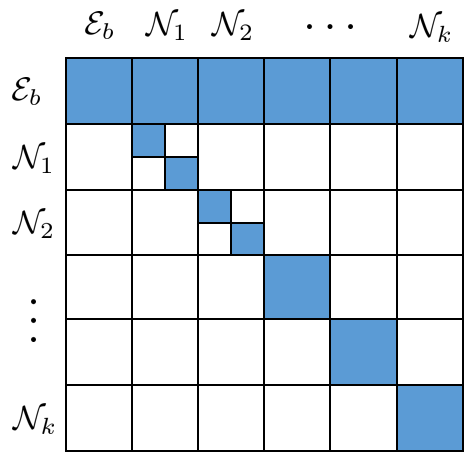}
}{
\includegraphics[width=.25\textwidth]{figs/matrix_block.png}
}
\caption{Non-zero entries of the $K_{e\hat{e}}$ matrix (as represented by the dark blocks) for a graph with tree partition $\set{\calN_1,\calN_2,\cdots,\calN_k}$ and bridge set $\calE_b$. The small blocks represent cells inside the regions.}\label{fig:matrix_block}
\end{figure}

\section{Non-Bridge Failures are Localizable}\label{section:non-bridge}
In this section, we characterize the power flow redistribution under the DC model when a non-bridge line is tripped and show that such failures are localized by the tree partition regions. More specifically, we study how the tripping of a line $e\in\calE$ impacts the branch flow on a different edge $\hat{e}\in\calE$. Recall, we mentioned in Section \ref{section:model} that when $e$ is not a bridge, the power flow change on $\hat{e}$ due to tripping $e$ is given by
\begin{equation*}
\Delta P_{\hat{e}} = K_{e\hat{e}} \times P_e.
\end{equation*}
The impact of the line failure of $e$ can thus be characterized by the distribution factor $K_{e\hat{e}}$.

\subsection{Impact across Regions}
To start with, we consider the case where $\hat{e}$ does belong to the same region as $e$, that is, $\hat{e}$ either belongs to a different region or $\hat{e}$ is a bridge. This case was studied in~\cite{guo2017monotonicity} where the following result was shown.

\begin{prop}\label{prop:disjoint_island}
Consider a power network $\calG$ with tree partition $\set{\calN_1,\calN_2,\cdots, \calN_k}$. Let $e, \hat{e}\in\calE$ be two different edges such that $e$ is not a bridge. Then,
$$
K_{e\hat{e}} = 0
$$
for any $\hat{e}$ that is not in the same the region containing $e$.
\end{prop}

This result implies that, when a non-bridge line $e$ fails, any line $\hat{e}$ not in the same tree partition region as $e$ will not be affected, regardless of whether $\hat{e}$ is a bridge. In other words, non-bridge failures cannot propagate through the boundaries formed by the tree partition regions of $\calG$.

\subsection{Impact within Regions}
It is reasonable, based on physical intuition, to expect that the converse to the above result is also true. That is, if $e, \hat{e}$ belong to a common region (and thus $e$ is not a bridge), we would expect $K_{e\hat{e}}\neq 0$. This, however, is not always the case for two reasons: (a) some vertices within a tree partition region may ``block'' spanning forests from $e$ to $\hat{e}$; (b) the graph $\calG$ may be too symmetric. We elaborate on these two scenarios separately in the following two subsections.

\subsubsection{Block Decomposition}
To illustrate the issue described above, we use the following example to demonstrate that certain vertices within a tree partition region may ``block'' spanning forests from $e$ to $\hat{e}$.

\begin{figure}
\centering
\iftoggle{isarxiv}{
\subfloat[]{\includegraphics[height=2cm]{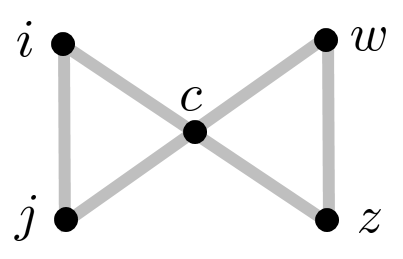}
}
\hfill
\subfloat[]{\includegraphics[height=2cm]{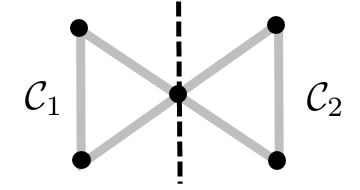}
}
}{
\subfloat[]{\includegraphics[height=2cm]{figs/butterfly.png}
}
\hfill
\subfloat[]{\includegraphics[height=2cm]{figs/butterfly_decomp.png}
}
}
\caption{(a) A butterfly network. (b) The block decomposition of the butterfly network into cells $\calC_1$ and $\calC_2$.}\label{fig:butterfly}
\end{figure}

\begin{eg}\label{eg:butterfly}
Consider a butterfly network shown in Fig.~\ref{fig:butterfly}(a) and pick $e=(i,j)$ and $\hat{e}=(w,z)$ from the butterfly wings. It is not hard to see that any tree containing $i,w$ simultaneously must contain the body vertex $c$, and so is any tree containing $j,z$. As a result, we see $\calT\paren{\set{i,w},\set{j,z}}$ is empty. Similarly we also know $\calT\paren{\set{i,z},\set{j,w}}$ is empty. By Theorem \ref{thm:redist_factor}, we then have $K_{e\hat{e}}=0$. 
\end{eg}

The issue with Example \ref{eg:butterfly} is that the butterfly graph is not 2-connected. In other words, it is possible that the removal of a single vertex (in this case the body vertex $c$) can disconnect the original graph. We refer to such a vertex as a \emph{cut vertex} following graph-theoretic convention. From Example \ref{eg:butterfly}, we see that cut vertices may ``block'' spanning forests between edges in a tree partition region. That is, two disjoint trees cannot pass through a common cut vertex.

Fortunately, we can precisely capture such an effect by decomposing each tree partition region further through the classical \emph{block decomposition}~\cite{harary1969graph}. Recall that the block decomposition of a graph is a partition of its \emph{edges} such that each partitioned component is 2-connected. See Fig.~\ref{fig:butterfly}(b) for an illustration. We refer to such components as \emph{cells} to reflect the fact that they are smaller parts within a tree partition region. Note that two different cells within a tree partition region may share a common vertex as the block decomposition is over graph edges. The block decomposition of a graph always exists and can be found in linear time~\cite{tarjan1985efficient}.

\begin{lemma}
Consider a power network $\calG$ and let $e,\hat{e}$ be two distinct edges within the same tree partition region but across different cells. Then $K_{e\hat{e}}=0$.
\end{lemma}
\begin{proof}
Let $e=(i,j)$ be within the cell $\calC_e$ and $\hat{e}=(w,z)$ be within the cell $\calC_{\hat{e}}$. It is a classical result that any path originating from a vertex within $\calC_e$ to a vertex within $\calC_{\hat{e}}$ must pass through a common cut vertex in $\calC_e$ \cite{harary1969graph}. As a result, it is impossible to find two disjoint trees in $\calG$ containing $i,w$ and $j,z$ respectively. Thus $\calT\paren{\set{i,w}, \set{j,z}}$ is empty. Similarly $\calT\paren{\set{i,z}, \set{j,w}}$ is also empty. By Theorem \ref{thm:redist_factor}, we then know $K_{e\hat{e}}=0$.
\end{proof}

\subsubsection{Symmetry}
Next, we demonstrate that graph symmetry\footnote{By symmetry, we refer to graph automorphisms. The exact meaning of symmetry, however, is not important for our purpose.} may also block the propagation of failures.  Again, we illustrate the issue with a simple example. 

\begin{eg}\label{eg:symmetry}
Consider the complete graph on $n$ vertices and pick $e=(i,j)$ and $\hat{e}=(w,z)$ such that $e$ and $\hat{e}$ do not share any common endpoints: $i\neq j\neq w\neq z$. Assume the line susceptances are all $1$. By symmetry, it is easy to see that there is a bijective correspondence between $\calT\paren{\set{i,w},\set{j,z}}$ and $\calT\paren{\set{i,z},\set{j,w}}$, and thus
$$
\sum_{E\in \calT(\set{i,w},\set{j,z})}\chi(E)-\sum_{E\in \calT(\set{i,z},\set{j,w})}\chi(E) = 0.
$$
By Theorem \ref{thm:redist_factor}, we then have $K_{e\hat{e}}=0$.
\end{eg}

A complete graph is 2-connected and thus forms a cell. Example \ref{eg:symmetry} shows that even if the two edges $e,\hat{e}$ are within the same cell, when the graph $\calG$ is rich in symmetries, it is still possible that a failure of $e$ does not impact $\hat{e}$. Nevertheless, this issue is not critical as such symmetry almost never happens in practical systems because of heterogeneity in line susceptances. In fact, even if the system is originally symmetric, an infinitesimal change on the line susceptances is enough to break the symmetry, as we now show.

More formally, we adopt a form of perturbation analysis on the line susceptances. That is, instead of requiring the line susceptance to be fixed values $B_e$, we add a random perturbation $\omega=(\omega_e:e\in\calE)$ drawn from a probability measure $\mu$. Such perturbations can come from manufacturing error or measurement noise. The perturbed system\footnote{We assume the perturbation ensures $B_e+\omega_e > 0$ for any $e\in\calE$ so that the new susceptance is physically meaningful.} shares the same topology (and thus tree partition) as the original system, yet admits perturbed susceptances $B+\omega$. The randomness of $\omega$ implies the factor $K_{e\hat{e}}$ is now a random variable. Let $\calL_m$ be the Lebesgue measure on $\R^m$. Recall $\mu$ is \emph{absolutely continuous} with respect to $\calL_m$ if for any measurable set $S$ such that $\calL_m(S)=0$, we have $\mu(S)=0$.

\begin{prop}\label{prop:same_region}
Consider a power network $\calG$ under perturbation $\mu$ and let $e,\hat{e}$ be two distinct edges within the same cell. If $\mu$ is absolutely continuous with respect to $\calL_m$,
$$
\mu\paren{K_{e\hat{e}}\neq 0} = 1.
$$
\end{prop}
\iftoggle{isreport}{
See Appendix \ref{proof:same_region} for a proof.
}{
See our online report for a proof \cite{report}.
}

Note that, by Radon-Nikodym theorem~\cite{royden2010real}, the probability measure $\mu$ is absolutely continuous with respect to $\calL_m$ if and only if it affords a probability density function. In other words, there are no requirements on either the power or the correlation of the perturbation for Proposition \ref{prop:same_region} to apply. The only necessary condition is that the measure $\mu$ cannot contain Dirac masses. As a result, we see that for almost all practical probability models of such perturbation (e.g.,~truncated Gaussian noise with arbitrary covariance, bounded uniform distribution, truncated Laplace distribution), $K_{e\hat{e}}\neq 0$ for $e,\hat{e}$ within the same cell almost surely, no matter how small the perturbation is. 

\section{Bridge Failures Propagate}\label{section:bridge}
The remaining case necessary to prove Theorem \ref{thm:summary} is a characterization of the power flow redistribution when a bridge is tripped.  Here, we show that such failures generally propagate through the entire network.

Recall that, when $e$ is not a bridge, the branch flow change on $\hat{e}$ due to tripping $e$ is given by
\begin{equation}\label{eqn:power_change}
\Delta P_{\hat{e}} = K_{e\hat{e}} \times P_e.
\end{equation}
When $e$ is a bridge, tripping $e$ disconnects the power grid into two components, and the power in each connected component may not be balanced. Such power imbalance can be resolved by a power balance rule $\calB$ (see \cite{bernstein2014power,bienstock2010n-k,bienstock2011control,bernstein2014vulnerability} for examples of such rules), which together with the DC model uniquely determines the new branch flows (and thus the branch flow change $\Delta P_{\hat{e}}$). We extend the definition of $K_{e\hat{e}}$ through \eqref{eqn:power_change} to the case where $e$ is a bridge and call it the \emph{extended line outage distribution factor}. Besides being related to the $B$ and $C$ matrices, this extended $K_{e\hat{e}}$ factor also depends on the power injection $p$ and the power balance rule $\calB$.

Power networks without microgrids typically operate in ``island-free mode'' as islanding (i.e.,~isolating a part of the grid power flow from the rest of the network) poses a safety hazard to utility maintainence and repair personnel and potentially leads to damage of the infrastructure \cite{ropp1999analysis}. Formally we define the concept of island-free as follows. 

\begin{defn}\label{defn:island_free}
For a power network $\calG$, an injection $p$ is said to be \textbf{island-free} if under the injection $p$, the branch flow $P$ in $\calG$ satisfies $P_e\neq0$ for any bridge $e$.
\end{defn}

Intuitively, island-free means that no part of the grid balances its own power.

 For a grid $\calG$ operating with the balance rule $\calB$, when a net power imbalance of level $M$ is detected, the rule $\calB$ selects a set of \emph{participating} buses from $\calG$ and adjust their injections properly to cancel $M$. The rules studied in the literature \cite{soltan2015analysis,bernstein2014power,bienstock2010n-k,bienstock2011control,bernstein2014vulnerability} are typically \emph{linear} in the sense that, for any participating bus $j$, the injection adjustment $\delta p_j$ dictated by the rule $\calB$ is linear in $M$.

Denote the set of participating buses of $\calB$ as $\calN_{\calB}$ and let $n_b=\abs{\calN_{\calB}}$. The rule $\calB$ can then be interpreted as a linear transformation from $\R$ to $\R^{n_b}$ given by
$$
\calB(M) = \paren{\alpha_j M: j\in\calN_{\calB}},
$$
where $\alpha_j$ are positive constants that sum to $1$. For instance, if the imbalance is uniformly absorbed by the generators as in \cite{bernstein2014power,soltan2015analysis}, we have $\calN_{\calB}$ to be the set of generators and $\alpha_j=1/G$, where $G$ is the number of generators. As another example, if the imbalance is regulated through Automatic Generation Control, then we have $\calN_{\calB}$ to be set of controllable generators and $\alpha_j$ are the normalized generator participation factors.

\begin{defn}\label{defn:participating}
For a power grid $\calG$ with tree partition $\calP=\set{\calN_1,\calN_2,\cdots, \calN_k}$ operating under power balance rule $\calB$, a region $\calN_i$ with block decomposition $\set{\calC_{1}^i, \calC_2^i, \cdots, \calC_{m_i}^i}$ is said to be a \textbf{participating region} if $\calN_{\calB}\cap \calC^i_j$ contains a non-cut-vertex for $j=1,2,\ldots, m_i$. The grid $\calG$ is said to be a \textbf{participating grid} if $\calN_i$ is participating for  $i=1,2,\ldots,k$.
\end{defn}

A power grid is usually participating. For instance, if $\calB$ allocates the power imbalance uniformly to the generators as in \cite{bernstein2014power,soltan2015analysis}, as long as each cell in the network contains a generator that is not at the ``gate'' (i.e., it is not a cut vertex\footnote{A cut vertex is always at the ``gate'' of a cell as all paths from outside towards vertices in the cell must pass through a cut vertex\cite{harary1969graph}.}), the grid is participating. If in addition load-side participation is exploited and $\calB$ also allocates the power imbalance to controllable loads, then the power grid is participating if the controllable loads are ubiquitously deployed to all buses. 

Given the above, we now state our main result.

\begin{prop}\label{prop:bridge_prop}
Consider a participating power network $\calG$ with island-free injection $p$ and under line susceptance perturbation $\mu$ that is absolutely continuous with respect to $\calL_m$. If $e$ is a bridge of $\calG$, then for any $\hat{e}\neq e$, we have
$$
\mu(K_{e\hat{e}}\neq 0)=1.
$$
\end{prop}
\iftoggle{isreport}{
See Appendix \ref{proof:bridge_prop} for a proof.
}{
See our online report for a proof \cite{report}.
}

The proof of Proposition \ref{prop:bridge_prop} provides interesting insights. In particular, both the participating grid and island-free injection conditions are in fact necessary. For instance, if the supply and demand is balanced within a region (which violates the island-free condition), then when the bridge connecting this region to other parts of the grid is tripped, the power is still balanced and thus the power flow within the region stays unchanged. This is precisely the case when a microgrid is connected to the power network: by disconnecting the microgrid from the main grid (which effectively trips a bridge in the original system), branch flows within both the microgrid and the main grid are not impacted.

Note that the result in Proposition \ref{prop:bridge_prop} can be easily extended to more general settings. It is straightforward to extend our proof to cover the case where neither condition is posed. We prefer not to present this generalization here because it would unnecessarily complicate the proof of the result without providing new insights.

\section{Localizing Cascading Failures}\label{section:case_study}
Our findings highlight a new approach for improving the robustness of the network. More specifically, Theorem~\ref{thm:summary} and the discussion in Section~\ref{section:non-bridge} suggest that it is possible to localize failure propagation by \emph{temporarily} switching off certain transmission lines based on the specific injection. This creates more, smaller areas where failure cascades can be contained.  We remark that the lines that are switched off are still part of the system. In cases where the newly created bridges are tripped, some of these lines should be switched on so the system are still connected.  The examples presented below are preliminary and we are still investigating how to optimally tradeoff the increased robustness from localized failures and the yet also increased vulnerability from having more bridges.
 
It is reasonable to expect that such an action may increase the stress on the remaining lines and, in this way, worsen the network congestion. In fact, one may expect that improved system robustness obtained by switching off lines \textit{always} comes at the price of increased congestion levels. In this section, we argue that this is not necessarily the case, and show that if the lines to switch off are selected properly, it is possible to \emph{improve the system robustness and reduce the congestion simultaneously.} We corroborate this claim by considering first a small stylized example and then an IEEE test system.

\subsection{Double-Ring Network}
\begin{figure}
\centering
\iftoggle{isarxiv}{
\subfloat[]{\includegraphics[width=.22\textwidth]{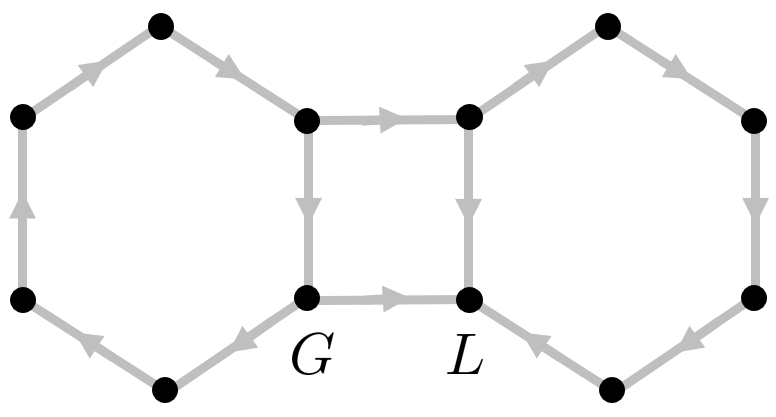}
}
\hspace{0.2cm}
\subfloat[]{\includegraphics[width=.22\textwidth]{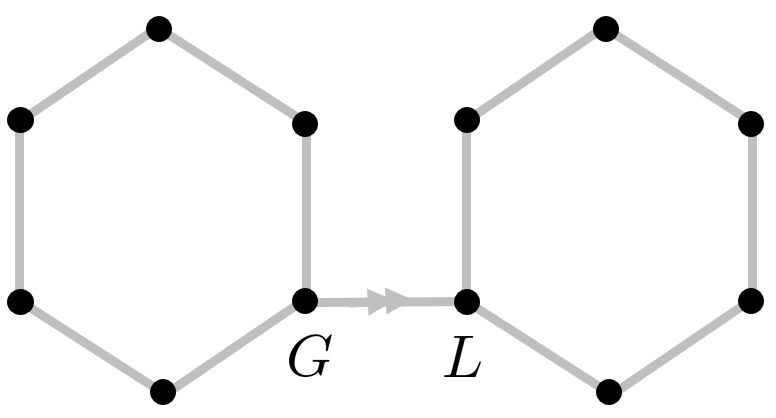}
}
}{
\subfloat[]{\includegraphics[width=.22\textwidth]{figs/double_ring.png}
}
\hspace{0.2cm}
\subfloat[]{\includegraphics[width=.22\textwidth]{figs/double_ring_trip.png}
}
}
\caption{(a) A double-ring network. $G$ is the generator bus and $L$ is the load bus. Arrows represent the original power flow. (b) The new network after removing an edge. Arrows represent the new power flow.}\label{fig:double_ring}
\end{figure}

Consider the double-ring network in Fig.~\ref{fig:double_ring}(a), which contains exactly one generator and one load bus. The original power flow on this network is also shown in Fig.~\ref{fig:double_ring}(a). Suppose we switch off the upper tie-line. The new network and the redistributed power flow are shown in Fig.~\ref{fig:double_ring}(b). In this example, by switching off one transmission line, the circulating flows inside the hexgons are removed and the overall network congestion is decreased. In fact, it is easy to show that the topology in Fig.~\ref{fig:double_ring}(b) minimizes the sum of (absolute) branch flows over all possible topologies.

\subsection{IEEE test system}

\begin{figure*}[t]
\centering
\iftoggle{isarxiv}{
\subfloat[Original influence graph.]{
\includegraphics[width=.83\textwidth]{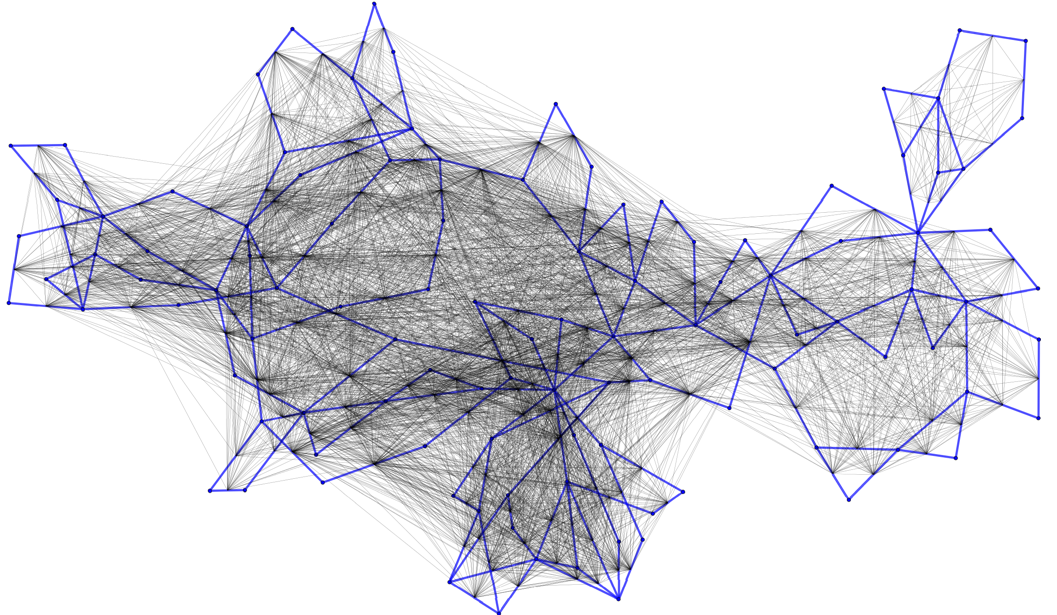}
}
\hfill
\subfloat[The influence graph after switching off $e_1$, $e_2$ and $e_3$. The black dashed line indicates the failure propagation boundary defined by the tree partition. The vertices $c_1$ and $c_2$ are cut vertices.]{
\includegraphics[width=.83\textwidth]{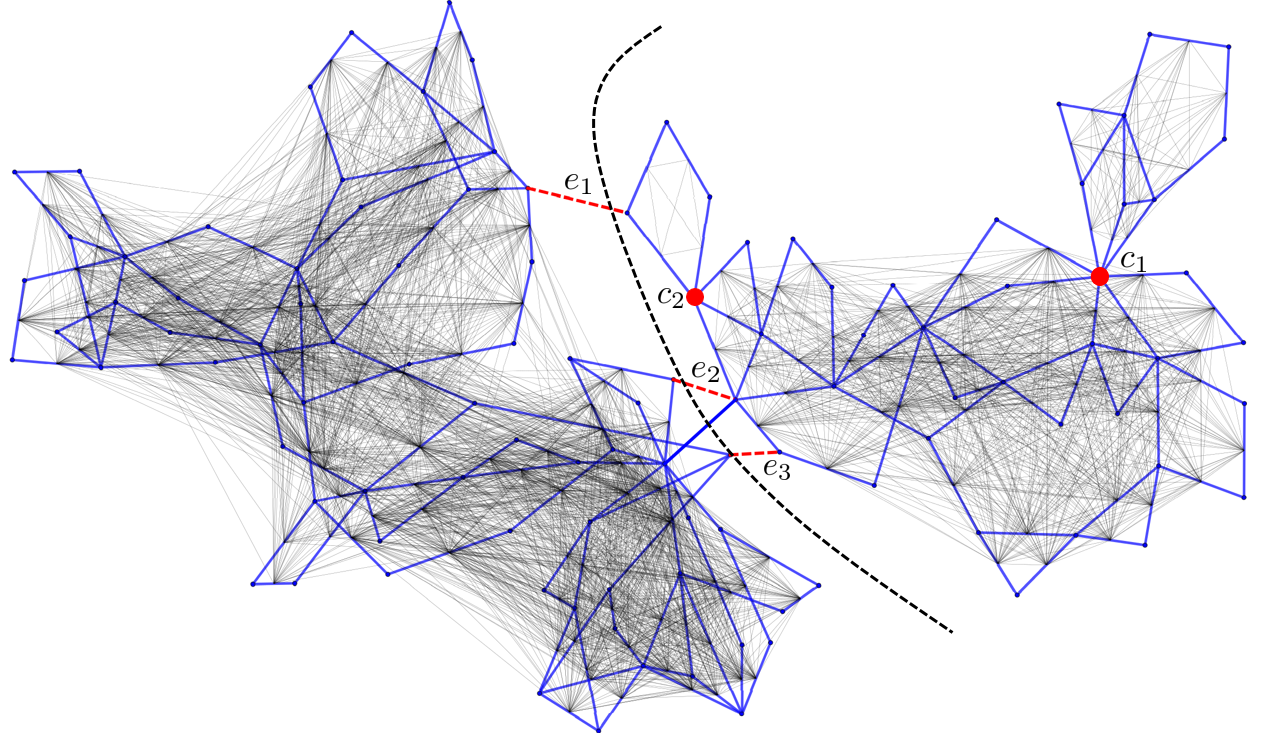}
}
}{
\subfloat[Original influence graph.]{
\includegraphics[width=.83\textwidth]{figs/inf_graph_before.png}
}
\hfill
\subfloat[The influence graph after switching off $e_1$, $e_2$ and $e_3$. The black dashed line indicates the failure propagation boundary defined by the tree partition. The vertices $c_1$ and $c_2$ are cut vertices.]{
\includegraphics[width=.83\textwidth]{figs/inf_graph_after.png}
}
}
\caption{Influence graphs on the IEEE 118-bus network before and after switching off lines $e_1$, $e_2$ and $e_3$. Blue edges represent physical transmission lines and grey edges represent connections in the influence graph.}\label{fig:inf_graph}
\end{figure*}

In the simple example above, removing a line provides improvements in both robustness and congestion. Now, we move to the case of a more realistic network, the IEEE 118-bus test system. In this case, we also see that line removals can improve robustness without more than minor increase in congestion.

In our experiments, the system parameters are taken from the Matpower Simulation Package \cite{zimmerman2011matpower} and we plot the influence graphs among the transmission lines to demonstrate how a line failure propagates in this network\footnote{The original IEEE 118-bus network has some trivial ``dangling'' bridges that we remove (collapsing their injections to the nearest bus) to obtain a more transparent influence graph.}. More specifically, in the influence graph we plot, two edges $e$ and $\hat{e}$ are connected if the impact of tripping $e$ on $\hat{e}$ is not negligible (we use $|K_{e\hat{e}}|\geq 0.005$ as a threshold). In Fig.~\ref{fig:inf_graph}(a), we plot the influence graph of the original network. It can be seen that this influence graph is very dense and connects many edges that are topologically far away, showing the non-local propagation of line failures within this network.

Next, we switch off three edges (indicated as $e_1$, $e_2$ and $e_3$ in Fig.~\ref{fig:inf_graph}(b)) to obtain a new topology that has a bridge and whose tree partition now consists of two regions of comparable size. The new influence graph is shown in Fig.~\ref{fig:inf_graph}(b). One can see that, compared to the original influence graph in Fig.~\ref{fig:inf_graph}(a), the new influence graph is much less dense and, in particular, there are no edges connecting transmission lines that belong to different tree partition regions. We also note that the network in Fig.~\ref{fig:inf_graph}(b) contains two cut vertices (indicated by $c_1$ and $c_2$ in the figure, with $c_2$ being created when we switch off the lines). It can be checked that line failures are ``blocked'' by these cut vertices, which verifies our results in Section \ref{section:non-bridge}.

It is also of interest to see how the network congestion is impacted by switching off these lines. To do so, we collect statistics on the difference between the branch flows in Fig.~\ref{fig:inf_graph}(b) and those in \ref{fig:inf_graph}(a). In Fig.~\ref{fig:lineflowscomparison}(a), we plot the histogram of such branch flow differences normalized by the original branch flow in Fig.~\ref{fig:inf_graph}(a). It shows that roughly half (the exact percentage is $47.41\%$) of the transmission lines have higher congestion yet the majority of these branch flow increases are negligible. To more clearly see how much the congestion becomes worse on these lines, we plot the cumulative distribution function of the normalized
positive branch flow changes, which is shown in Fig.~\ref{fig:lineflowscomparison}(b). One can see from the figure that $90\%$ of the the branch flows increase by no more than $10\%$. 

\begin{figure}
\centering
\iftoggle{isarxiv}{
\subfloat[]{\includegraphics[width=.24\textwidth]{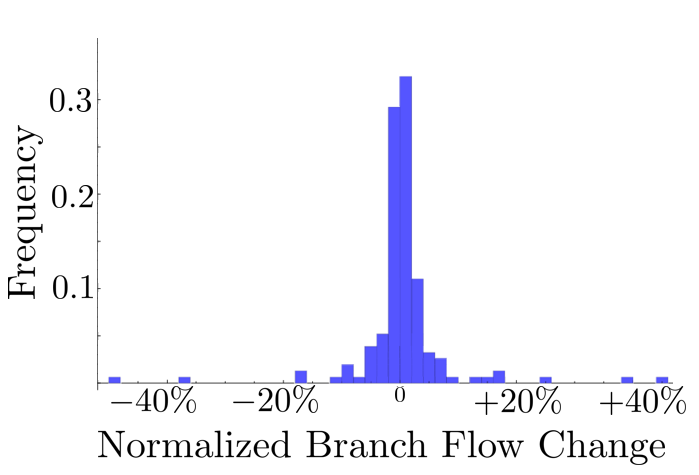}
}
\subfloat[]{\includegraphics[width=.24\textwidth]{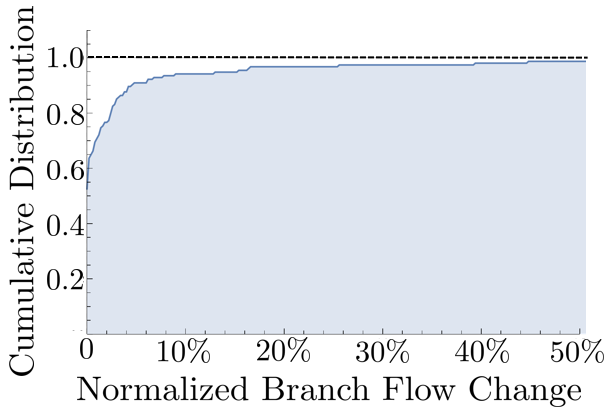}
}
}{
\subfloat[]{\includegraphics[width=.24\textwidth]{figs/histogram.png}
}
\subfloat[]{\includegraphics[width=.24\textwidth]{figs/cdf.png}
}
}
\caption{(a) Histogram of the normalized branch flow changes. (b) Cumulative distribution function of the positive normalized branch flow changes. Note that the curve intercepts the $y$-axis since $52.59\%$ of the branch flows decrease.}
	\label{fig:lineflowscomparison}
\end{figure}

\section{Conclusion}\label{section:conclusion}
In this work, we provide an analytical characterization of line failure localizability in power systems. We demonstrate that the transmission network graph encodes rich information on the regions that a line failure can impact and that such regions can be computed in linear time. Further, using a case study on the IEEE 118-bus test network, we show that switching off certain transmission lines can improve the grid robustness against line failures without significantly increasing line congestion.

This work can be extended in several directions. First, we provide an analytical characterization of power flow redistribution when a line fails, and our results are generalizable to bus failures. It is of interest to understand how these two types of failures interact. Second, we demonstrate in our case studies that by switching off certain transmission lines, grid robustness can potentially be improved. It would be useful if the selection of such lines can be optimized for a certain objective function, such as the sparsity of the influence graph or the total load loss when some critical lines are tripped. Third, to fully capture the cascading failure dynamics, both the power flow redistribution and the line capacities are relevant. It is important to investigate how line capacities can be incorporated to our framework.

\bibliographystyle{IEEEtran}
\bibliography{biblio}

\iftoggle{isreport}{
\appendices
\section{Proof of Proposition \ref{prop:unique_irreducible}}\label{proof:uniqueness}
For $\calG=(\calN, \calE)$, let $\calE_b$ be the set of edges that all spanning trees of $\calG$ pass through. In other words, $\calE_b$ is the set of all bridges in classical graph theory. Let $\calG_b=(\calN, \calE\bs\calE_b)$ denote the graph obtained from $\calG$ by removing all edges in $\calE_b$, and let $\calP^*=\set{C_1, C_2, \cdots, C_{k}}$ be its connected components, where $k=|\calE_b|+1$. We claim $\calP^*$ is a irreducible tree partition of $\calG$.

First we show that the reduced graph $\calG_{\calP^*}$ is a tree. Assume not, then there is a loop in $\calG_{\calP^*}$. Without loss of generality, let us assume the loop is $\paren{C_{1}, C_{2}, \cdots, C_{l}}$ for some $2\le l \le k$. Then by the way we form $\calG_{\calP^*}$, there exist vertices $n_1^t, n_1^s, n_2^t, n_2^s, n_3^t, \cdots, n_l^t, n_l^s$ such that:
\begin{enumerate}
\item For each $i$, the vertices $n_i^s,n_i^t\in C_i$
\item For each $i$, the edge $e_{i,i+_l1}:=(n_i^s, n_{i+_l1}^t) \in \calE$, where $+_l$ denotes the addition modulo $l$
\end{enumerate}
For any $i$, since $C_i$ is connected, we can find a path $P_i$ from $n_i^t$ to $n_i^s$. It is then clear that the concatenated path
$$
\paren{P_1,e_{1,2},P_2,e_{2,3}\ldots e_{l-1,l},P_l,e_{l,1}}
$$
forms a loop in the original graph $\calG$. As a result, not all spanning trees pass through $e_{1,2}$ and thus $e_{1,2}\notin\calE_b$, which leads to a contradiction.

Next we show $\calP^*$ is irreducible. To do so, we prove that $\calP^*$ is finer than any tree partition $\calP:=\set{\calN_1, \calN_2, \cdots, \calN_{k'}}$ of $\calG$. Consider a region in $\calP^*$, say $C_1$. Since both $\calP^*$ and $\calP$ are partitions of $\calG$, there must be some region in $\calP$, say $\calN_1$, such that $C_1\cap \calN_1\neq \emptyset$. We claim $C_1\subset \calN_1$. Otherwise, there exists another region in $\calP$, say $\calN_2$, such that $C_1\cap \calN_2\neq \emptyset$. Pick $n_1\in C_1\cap \calN_1$ and $n_2\in C_1\cap \calN_2$. Then $n_1\neq n_2$ because $\calN_1\cap\calN_2=\emptyset$. Now since $n_1,n_2\in C_1$, and $C_1$ does not contain any bridge (in classical graph theory sense), by Menger's Theorem \cite{menger1927allgemeinen}, there exists a cycle (which is not necessarily simple\footnote{A cycle is simple if it visits each vertex at most once.}) in $C_1$ containing both $n_1$ and $n_2$. By collapsing adjacent vertices in this cycle that belong to common regions, we can find regions $\calN_{l_1}^1, \calN_{l_2}^1, \cdots, \calN_{l_{p_1}}^1, \calN_{l_1}^2, \calN_{l_2}^2, \cdots, \calN_{l_{p_2}}^2$ so that the  path from $n_1$ to $n_2$ in this cycle is given by
$$
\paren{P^1_1,e_{1,l_1}^1,P^1_{l_1}, e_{l_1,l_2}^1,  \cdots, e_{l_{p_1}, 2}^1, P^1_2}
$$
and the path from $n_2$ to $n_1$ in this cycle is given by
$$
\paren{P^2_2,e_{2,l_1}^2,P^2_{l_1}, e_{l_1,l_2}^2,  \cdots, e_{l_{p_2}, 1}^2, P^2_1}
$$
where $e_{i, j}^{1}, e_{i,j}^{2}$ are edges with source vertices in $\calN_i$ and target vertices in $\calN_j$ and $P_i^1, P_i^2$ are paths contained in $\calN_i$. As a result, we see
$$
\paren{\calN_1,\calN_{l_1}^1, \cdots, \calN_{l_{p_1}}^1, \calN_2, \calN_{l_1}^2, \cdots, \calN_{l_{p_2}}^2}
$$
forms a loop in $\calG_{\calP}$. This implies $\calP$ is not a tree partition, contradicting our assumption.

We thus have shown that $\calP^*$ is a irreducible tree partition of $\calG$. Moreover, for any other irreducible tree partition $\ol{\calP}$, the above proof shows that
$$
\calP^*\succeq \ol{\calP}
$$
Since $\ol{\calP}$ is irreducible and thus maximal with respect to $\succeq$, we see $\ol{\calP} = \calP^*$. In other words, any irreducible tree partition of $\calG$ must coincide with $\calP^*$. This completes our proof. \qedd
\section{Proof of Proposition \ref{prop:same_region}}\label{proof:same_region}
Let $\calC$ be the cell that $e$ and $\hat{e}$ belong to and write $e=(i,j)$ and $\hat{e}=(w,z)$. Consider the polynomial in line susceptances $(B_e:e\in\calE)$ defined as
$$
f(B):=\sum_{E\in\calT\paren{\set{i,w}, \set{j,z}}}\chi(E)-\sum_{E\in\calT\paren{\set{i,z}, \set{j,w}}}\chi(E)
$$
We claim $f$ is not identically zero.

First, let $C$ be a simple cycle in $\calC$ that contains both $e$ and $\hat{e}$. Such a cycle always exists as $\calC$ is 2-connected by construction, and it is a classical result that any pair of edges are contained in a simple cycle for 2-connected graphs \cite{harary1969graph}. Therefore, we can find two disjoint paths $P_1$ and $P_2$ connecting the endpoints of $e$ and $\hat{e}$. Without loss of generality, assume $P_1$ connects $i$ to $w$ and $P_2$ connects $j$ to $z$. By iteratively adding edges from $\calG$ to $P_1$ and $P_2$, we can extend $P_1$ and $P_2$ to a spanning forest of $\calG$ consisting of exactly two trees. Moreover, the tree extended from $P_1$ contains $i, w$ and the tree extended from $P_2$ contains $j,z$. We thus have constructed an element of $\calT\paren{\set{i,w}, \set{j,z}}$. Denote this element by $E_0$.

Second, we show that 
$$\calT\paren{\set{i,w}, \set{j,z}}\cap \calT\paren{\set{i,z}, \set{j,w}}=\emptyset$$
Indeed, consider an element $E_1$ from $\calT\paren{\set{i,w}, \set{j,z}}$, which consists of two trees $\calT_1$ and $\calT_2$ with $\calT_1$ containing $i,w$ and $\calT_2$ containing $j,z$. If $E_1\in\calT\paren{\set{i,z}, \set{j,w}}$, then $\calT_1$ must also contain $z$. However, this implies $z\in \calT_1\cap \calT_2$, and thus $\calT_1$ and $\calT_2$ are not disjoint, contradicting the definition of $\calT\paren{\set{i,w}, \set{j,z}}$.

As a result, we see that the element $E_0$ constructed in our first step contributes a term to $\sum_{E\in\calT\paren{\set{i,w}, \set{j,z}}}\chi(E)$ but not to $\sum_{E\in\calT\paren{\set{i,z}, \set{j,w}}}\chi(E)$. Therefore $f(B)$ contains nonvanishing terms and is not identically zero.

It is well-known from algebraic geometry that the root set of a polynomial which is not identically zero has Lebesgue measure zero \cite{federer1969geometric}. That is, we have
$$
\mu\paren{f(B+\omega)=0} = \calL_m\paren{f(B+\omega)=0} = 0
$$
where the first equality is because $\mu$ is absolutely continuous with respect to $\calL_m$ (it is clear that the root set of the polynomial $f$ is measurable).

Finally, by Theorem \ref{thm:redist_factor}, we know $K_{e\hat{e}}= 0$ if and only if $f(B+\omega)=0$. This then implies
$$
\mu(K_{e\hat{e}}\neq 0) = 1 - \mu(K_{e\hat{e}}=0) = 1-\mu\paren{f(B+\omega)=0} = 1
$$
and completes our proof. \qedd
\section{Proof of Proposition \ref{prop:bridge_prop}}\label{proof:bridge_prop}
\begin{figure}
\centering
\iftoggle{isarxiv}{
\includegraphics[width=.35\textwidth]{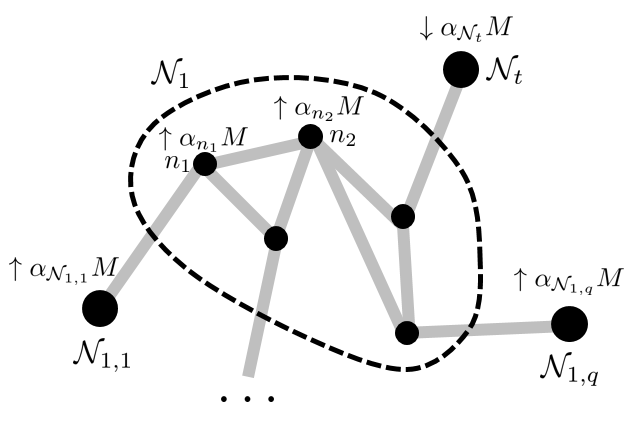}
}{
\includegraphics[width=.35\textwidth]{figs/localized_graph.png}
}
\caption{The localized graph $\calG_{\calN_1}$. $\calN_t$ is the imaginary bus containing $e$ and $\calN_{1,i}$'s are remaining imaginary buses. The power adjustments from the power balance rule $\calB$ are shown near each participating bus in reaction to a power loss of $M$.}\label{fig:localized}
\end{figure}
Denote the two connected components of $\calG$ after $e$ is tripped as $D_1,D_2$. Since the injection is island-free, we know $\sum_{j\in D_1}p_i\neq 0$ and $\sum_{j\in D_2}p_i\neq 0$, and thus the power is not balanced in neither $D_1$ nor $D_2$. We will show that $K_{e\hat{e}}\neq 0$ for any $\hat{e}$ within $D_1$. The case where $\hat{e}$ belongs to $D_2$ can be addressed similarly.

First we consider the case where $\hat{e}$ is a bridge. Denote the two connected components of $D_1$ after removing $\hat{e}$ as $D_{1,1}$ and $D_{1,2}$, and without loss of generality assume $D_{1,2}$ is originally connected to $e$ in $\calG$. It is easy to see that the branch flow change on $\hat{e}$ is given by
$$
\Delta P_{\hat{e}} = \sum_{j\in D_{1,1}}\delta p_j \neq 0
$$
where the last $\neq$ is because the grid is participating and thus all tree partition regions in $D_{1,1}$ would adjust their injections (in the same ``direction'' as $\alpha_j$'s are positive). As a result we see $K_{e\hat{e}}=\Delta P_{\hat{e}}/P_e\neq 0$.

Next we consider the case where $\hat{e}$ is not a bridge. In this case, $\hat{e}$ belongs to a certain tree partition region, say $\calN_1$. Let $n_1,n_2,\cdots, n_l$ be the set of participating buses in $\calN_1$. By collapsing all regions other than $\calN_1$, we can define a localized graph $\calG_{\calN_1}$ centered around $\calN_1$ as shown in Fig.~\ref{fig:localized}. Specifically, for each bridge $b$ incident to $\calN_1$, we create an imaginary bus that aggregates all injections inside all tree partition regions that can be reached by $\calN_1$ through $b$ \emph{before} $e$ is tripped, and this imaginary bus is connected to $\calN_1$ via the corresponding bridge $b$. If $e$ is not directly incident to $\calN_1$, then $e$ must connect two regions that are collapsed to a common imaginary bus, say $\calN_t$. If $e$ is incident to $\calN_1$ and its end point in $\calN_1$ is $n_e$, then we let $\calN_t$ be an additional imaginary bus that is connected to $n_e$ to mimic the edge $e$ (that is, before $e$ is tripped, this imaginary bus $\calN_t$ supplies power $P_e$ towards $n_e$ and after $e$ is tripped, $\calN_t$ loses all its injection; moreover the edge connecting $\calN_t$ to $n_e$ has susceptance $B_e$). Denote other imaginary buses as $\calN_{1,1},\calN_{1,2},\cdots, \calN_{1,q}$. 

Without loss of generality, let us assume because of the tripping of $e$, the aggregate power in $D_1$ is in shortage of $M$. Then the enforcement of the power balance rule $\calB$ would increase the power injection at each participating bus in $D_1$, which translates to the power adjustment as demonstrated in Fig.~\ref{fig:localized}. Specifically, the injection at $\calN_t$ would drop by $\alpha_{\calN_t}M$ as the power flow from $e$ is lost (the drop is generally not $M$ unless $e$ is directly incident to $\calN_1$); the injections at $\calN_{1,1},\calN_{1,2}, \cdots, \calN_{1,q}$ increase by $\alpha_{\calN_{1,i}} M$; and injections at the participating buses $n_1,n_2,\cdots, n_l$ in $\calN_1$ increase by $\alpha_{n_i}M$. In other words, by rebalancing power according to the rule $\calB$, we effectively shift the injections from $\calN_t$ to $\calN_{1,1},\calN_{1,2}, \cdots, \calN_{1,q}$ and $n_1,n_2,\cdots, n_l$.

Let the set of edges in $\calG_{\calN_1}$ be $\calE_1$. Pick $\calN_t$ to be the slack bus in this localized graph $\calG_{\calN_1}$ and define an index set
$$
\calI:=\set{\calN_{1,1},\calN_{1,2}, \cdots, \calN_{1,q}, n_1,n_2,\cdots, n_l}
$$
For any $i\in\calI$, let $g^{i}(B)$ be the following polynomial in line susceptances $(B_e:e\in\calE_1)$ 
$$
\sum_{E\in\calT_{\calE_1}\paren{\set{i,w},\set{\calN_t,z}}}\chi(E) - \sum_{E\in\calT_{\calE_1}\paren{\set{i,z},\set{\calN_t,w}}}\chi(E)
$$
where $\calT_{\calE_1}\paren{\calN_1,\calN_2}$ is the set of spanning forests of $\calG_{\calN_1}$ consisting of exactly two trees that contain $\calN_1$ and $\calN_2$ respectively (see Section \ref{section:model} for more details on the corresponding notion in $\calG$). By Proposition V.3 from our previous work \cite{guo2017monotonicity}, the branch flow change on $\hat{e}$ coming from the power shift from $\calN_t$ with amount $\alpha_iM$ towards $i$ is given by
$$
\Delta P^{i}_{\hat{e}}=\alpha_{i}M \times \frac{g^{i}(E)}{\sum_{E\in\calT_{\calE_1}}\chi(E)}
$$
where $\calT_{\calE_1}$ denotes the set of all spanning trees of $\calG_{\calN_1}$. Put
$$
g(B):=\sum_{i\in\calI}\alpha_{i}g^{i}(B)
$$
By linearity, we know
$$
\Delta P_{\hat{e}}=M\times \frac{g(B)}{\sum_{E\in\calT_\calE}\chi(E)}
$$

We claim $g(B)$ is not identically zero. Indeed, by a similar argument to the proof of Proposition \ref{prop:same_region}, we know that for any $i,j\in\calI$ (including the case $i=j$), the following is true:
$$
\calT_{\calE_1}\paren{\set{i,w},\set{\calN_t,z}}\cap \calT_{\calE_1}\paren{\set{j,z},\set{\calN_t,w}}=\emptyset
$$
As a result, a term in $g^i(B)$ with coefficient $1$ is never canceled by a term in $g^j(B)$ with coefficient $-1$, and vice versa. Therefore, to show $g(B)$ is not identically zero, it suffices to show at least one of $g^i(B)$ is not identically zero.

To do so, let $\calC$ be the cell that contains $\hat{e}$. Since $\calN_1$ is participating with respect to $\calB$, we know there exists a bus within $\calC$, say $n_1$, that participates the power balance and is not a cut vertex. Recall that any path from $\calN_t$ to $\calC$ must go through a common cut vertex in $\calC$ \cite{harary1969graph}, say $n_e$. Now by adding an edge between $n_e$ and $n_1$ (if it does not exist originally), the resulting cell $\calC'$ is still 2-connected. Thus there exists a simple loop in $\calC'$ that contains the edge $(n_e,n_1)$ and $\hat{e}=(w,z)$, which implies we can find two disjoint paths $P_1$ and $P_2$ connecting the endpoints of these two edges. Without loss of generality, assume $P_1$ connects $n_e$ to $w$ and $P_2$ connects $n_1$ to $z$. By concatenating the path from $\calN_t$ to $n_e$, we can extend $P_1$ to a path $\tilde{P}_1$ from $\calN_t$ to $w$, which is still disjoint from $P_2$. Now, by iteratively adding edges, we can extend $\tilde{P}_1$ and $P_2$ to a spanning forest of $\calG_{\calN_1}$ consisting of exactly two trees, which is an element of $\calT_{\calE_1}\paren{\set{n_1,z},\set{\calN_t,w}}$. This in particular implies $g^{n_1}(B)$ is not identically zero and therefore $g(B)$ is not identically zero.

Again by the classical algebraic geometry result asserting the root set of any polynomial that is not identically zero has Lebesgue measure zero\cite{federer1969geometric}, and because of the absolute continuity of $\mu$, we know
$$
\mu\paren{\Delta P_{\hat{e}}=0}=\calL_m\paren{g(B+\omega)=0}=0
$$
and thus
$$
\mu\paren{K_{e\hat{e}}\neq 0}=\mu\paren{\Delta P_{\hat{e}}\neq 0} = 1
$$
This completes our proof. \qedd
}{}

\end{document}